\documentclass[a4paper, USenglish, cleveref]{lipics-v2021}

\usepackage{tikz}
\usetikzlibrary{positioning,arrows.meta,calc,fit}
\usepackage{booktabs}
\usepackage{tabularx}
\usepackage{algorithm}
\usepackage[noend]{algpseudocode}
\usepackage{url}
\usepackage{xspace}
\usepackage[normalem]{ulem}

\definecolor{algkwcolor}{RGB}{0,90,140}  
\newcommand{\algkw}[1]{{\color{algkwcolor}\textbf{#1}}}
\definecolor{algcmtcolor}{RGB}{100,120,100}  
\newcommand{\algcmt}[1]{{\color{algcmtcolor}#1}}
\algrenewcommand{\algorithmiccomment}[1]{\hfill\algcmt{$\triangleright$ #1}}
\definecolor{algvarcolor}{RGB}{150,50,30}  
\newcommand{\V}[1]{{\color{algvarcolor}\mathit{#1}}}
\algrenewcommand{\alglinenumber}[1]{\color{algcmtcolor}\footnotesize #1}

\algrenewcommand\algorithmicdo{\algkw{do}}
\algrenewcommand\algorithmicif{\algkw{if}}
\algrenewcommand\algorithmicthen{\algkw{then}}
\algrenewcommand\algorithmicelse{\algkw{else}}
\algrenewcommand\algorithmicfor{\algkw{for}}
\algrenewcommand\algorithmicfunction{\algkw{function}}
\algrenewcommand\algorithmicprocedure{\algkw{procedure}}
\algrenewcommand\algorithmicreturn{\algkw{return}}
\algrenewcommand\Return{\State \algorithmicreturn{} }

\makeatletter
\def\@EventEmpty{}
\def\@EventStart#1#2#3{\algkw{upon}\ #1$($#2$)$\def\@EventTemp{#3}\ifx\@EventTemp\@EventEmpty\else\ #3\fi\ \algorithmicdo}
\algdef{SE}[EVENT]{Event}{EndEvent}[3]{\@EventStart{#1}{#2}{#3}}{\algorithmicend\ \algkw{event}}
\makeatother

\algdef{SE}[INIT]{Init}{EndInit}[1]{\algkw{upon initialization}\ #1\ \algorithmicdo}{\algorithmicend\ \algkw{init}}

\algdef{SE}[RECV]{Recv}{EndRecv}[2]{\algkw{upon}\ #1\ \algkw{from}\ #2\ \algorithmicdo}{\algorithmicend\ \algkw{recv}}

\algdef{SE}[CONS]{Cons}{EndCons}[1]{\algkw{upon}\ #1\ \algkw{from} consensus \algorithmicdo}{\algorithmicend\ \algkw{cons}}

\algdef{SE}[COND]{Cond}{EndCond}[1]{\algkw{when}\ #1\ \algorithmicdo}{\algorithmicend\ \algkw{cond}}

\algtext*{EndEvent}
\algtext*{EndInit}
\algtext*{EndRecv}
\algtext*{EndCons}
\algtext*{EndCond}

\newcommand{\gossiptext}[2]{\algkw{gossip} #1 \algkw{to} #2}
\algnewcommand\Gossip[2]{\State \gossiptext{#1}{#2}}

\newcommand{\broadcasttext}[1]{\algkw{broadcast} #1}
\algnewcommand\Broadcast[1]{\State \broadcasttext{#1}}

\newcommand{\proposetext}[2]{\algkw{propose}(#1, #2)}
\algnewcommand\Propose[2]{\State \proposetext{#1}{#2}}

\algnewcommand\Trigger[2]{\State \algkw{trigger} #1$($#2$)$}
\algnewcommand\SetTimeout[2]{\State \algkw{setTimeout}(#1, #2)}
\algnewcommand\InitGossip[0]{\State $gossipChannel \gets \text{getGossipChannel}()$}
\algnewcommand\InitConsensus[0]{\State $\text{initConsensus}()$}
\algnewcommand\Space[0]{\Statex}

\newcommand{\protoname}{\textsf{AMP}\xspace}

\newcommand{\payload}{payload\xspace}
\newcommand{\payloads}{payloads\xspace}

\newcommand{\collector}{proposer\xspace}

\newcommand{\collectors}{proposers\xspace}
\newcommand{\Collectors}{Proposers\xspace}

\title{\protoname: Arc Multi-Proposer Protocol with Bounded Inclusion Guarantees}
\titlerunning{\protoname: Arc Multi-Proposer Protocol}

\authorrunning{D. Cason, P. Vander Vos, J. Sousa, A. Seredinschi, and G. Liao}

\Copyright{Daniel Cason, Preston Vander Vos, João Sousa, Adi Seredinschi, and Gordon Liao}



\author{Daniel Cason}{Circle, Switzerland}{daniel.cason@circle.com}{0000-0003-1598-106X}{}
\author{Preston Vander Vos}{Circle, USA}{preston.vandervos@circle.com}{0009-0004-7206-5839}{}
\author{João Sousa}{Circle, Portugal}{joao.sousa@circle.com}{}{}
\author{Adi Seredinschi}{Circle, Switzerland}{adi.seredinschi@circle.com}{0000-0003-3648-3510}{}
\author{Gordon Liao}{Circle, USA}{gordon.liao@circle.com}{}{}

\nolinenumbers
\hideLIPIcs

\begin{document}
\pdfoutput=1

\maketitle

\begin{abstract}
    Blockchain systems that settle financial transactions
face a structural tension:
the validator that assembles each block holds unilateral power over
transaction inclusion and ordering.
Traditional markets curb this power through laws that prevent front-running and
market manipulation.
Regulators have flagged the absence of such rules as a first-order concern
for blockchain-based financial infrastructure.
To address this tension, we introduce \protoname, Arc multi-proposer protocol where no
single validator fully controls the flow of transactions into blocks.
\protoname composes with a broad class of BFT consensus algorithms in which
explicit validator votes determine decisions.
Any node can act as a \collector, collecting user transactions, bundling them
into payloads, and broadcasting those payloads to all validators.
Validators attest to received payloads by embedding payload identifiers in the
votes cast to reach a consensus decision.
This yields the key guarantee, \emph{bounded inclusion}: any payload
attested by more than $2f$ validators at height $h$ must appear in the block
finalized at height $h+1$.
A deterministic ordering function over finalized payloads curbs
any single validator's ordering discretion.
\protoname decouples dissemination from agreement, inherits safety and
liveness from the underlying BFT algorithm, and removes the need for a
shared transaction pool (mempool).
We prove correctness formally, including the bounded-inclusion guarantee within
one consensus height.

\end{abstract}


\section{Introduction}



In blockchains built on proposer-based Byzantine Fault Tolerant (BFT) consensus,
a single validator produces each block.
This validator, the \emph{block assembler}, determines which transactions are
included in the block and in what order.
The block assembler also acts as the consensus leader, the process that
coordinates agreement among validators.
This design, inherited from classic BFT protocols~\cite{castro1999pbft,
buchman_kwon_2016, yin2019hotstuff}, is simple and efficient: the leader
proposes values and validators validate the leader's proposal.
But it concentrates two powers in one actor: the power to \emph{exclude}
transactions and the power to \emph{order} them.

The exclusion power enables \emph{censorship} because the block assembler can
defer or omit transactions at will.
The ordering power enables \emph{maximal extractable value} (MEV), where the
block assembler reorders or inserts transactions to extract profit at users'
expense~\cite{Daian2020FlashBoys, fox2023censorship}.
In traditional financial markets, analogous behavior is prohibited by
regulation.
The Bank for International Settlements has noted that block assemblers ``can
engage in activities that would be illegal in traditional markets such as
front-running and sandwich trades''~\cite{bis2022mev}.
For blockchains used to settle institutional-grade financial flows, unconstrained power
over transaction inclusion and ordering is a market-integrity concern.


Rotating leaders is a common strategy to limit the power of the block
assembler, in particular to mitigate censorship.
If one block assembler omits a
transaction, the next may include it.
This approach became common in blockchain protocols after
Tendermint~\cite{buchman_kwon_2016}, replacing the stable-leader model of
PBFT~\cite{castro1999pbft}.
Yet censorship is hard to prove, as a single assembler can silently omit a
transaction without violating the protocol, and if multiple assemblers do
the same, the transaction can be postponed indefinitely.
Beyond censorship, rotating-leader blockchain designs typically rely on a
shared pool of pending transactions (a mempool).
Since any validator may become
the next assembler, pending transactions are disseminated to all validators
in advance.
This gives every validator visibility into transactions before they are
ordered, expanding the MEV attack surface to the entire validator set.


The single-assembler design also limits performance in blockchain systems.
First, throughput is capped by the assembler's egress bandwidth, leaving
the aggregate capacity of the remaining validators underused.
Second, throughput and latency are coupled: since consensus runs in lock-step,
increasing block size does not necessarily improve throughput, as larger blocks
take longer to propagate.
Third, in mempool-based designs, transactions are disseminated twice: first
through mempool gossip, then again as part of the proposed block.


To address the power concentration and performance limits of the
single-assembler design, we introduce
\protoname: \emph{Arc multi-proposer protocol}\footnote{\protoname originated
as an internal project at Circle for the Arc blockchain~\cite{liao_mayer_soghoian_jain_tierney_2025_arc}. During its design, we
recognized that the approach generalizes to a broad class of BFT consensus
algorithms and blockchains.}
that constrains validators' power over transaction inclusion
and ordering.
In \protoname, a set of nodes acts as \emph{\collectors}: they collect user
transactions, bundle them into \emph{\payloads}, and broadcast \payloads to all
validators.
Validators \emph{attest} to received \payloads by attaching \payload identifiers
to consensus votes that may lead to a decision.
Since a decision requires consensus votes from $n-f$ validators, any decision
necessarily carries every \payload identifier attached by more than $2f$
validators.
Since a valid commit certificate carries at least $n-f$ attestations,
a \payload attested by more than $2f$ validators survives up to $f$ Byzantine
omissions and $f$ assembler exclusions.
This yields \protoname's key guarantee, \emph{bounded inclusion}:
any \payload attested by more than $2f$ validators at height $h$
\emph{must} appear in the block finalized at height $h+1$.
The block finalized at a height is assembled by a deterministic function
that sequences transactions across all \payloads decided at that height,
removing validators' discretion over execution order.


\protoname decouples dissemination from agreement, removing the need for a
shared mempool: \collectors broadcast \payloads to all validators, while
consensus runs over \payload identifiers.
Running consensus over transaction identifiers, or over identifiers of
transaction batches, is a well-studied approach to improve
throughput~\cite{biely2012spaxos,danezis2022narwhal,spiegelman2022bullshark,
yang2022dispersedledger,babel2023mysticeti,spiegelman2023shoal,
giridharan2024autobahn,zarun2024shoalpp,bluebottle}.
Yet in many such designs, a single proposer still defines which identifiers
enter the consensus value, and thus the power to exclude transactions remains
concentrated~\cite{biely2012spaxos,danezis2022narwhal,spiegelman2022bullshark,
babel2023mysticeti,spiegelman2023shoal,zarun2024shoalpp,bluebottle}.
Protocols that do address this power often require either replacing the
consensus layer entirely~\cite{miller2016honeybadger,yang2022dispersedledger,
xiang2026prefix}
or adopting substantially more specialized protocol machinery~\cite{giridharan2024autobahn,
cadence, constellation}.


In contrast, \protoname composes with a broad class of BFT consensus algorithms
in which explicit validator votes determine decisions---it augments rather than
replaces the underlying consensus protocol.
\protoname requires the consensus protocol to support \emph{vote extensions}:
validators attach application-defined data to votes that may lead to a decision
and verify the extensions attached by other validators; validators drop votes
with invalid extensions.
At each height, a verifiable \emph{commit certificate}, containing votes from at
least $n-f$ validators, establishes the consensus decision.
\protoname requires the consensus protocol to expose the commit certificate for
each decision, as the \payload identifiers in its vote extensions determine
which \payloads are finalized at the next height.
Vote extensions were introduced by a Tendermint
implementation~\cite{cometbft_voteext_2026},
and are already deployed in production blockchains~\cite{buchman_kwon_2016,
polygon_heimdall_2026}.

We summarize our contributions:

\begin{enumerate}
  \item We introduce \protoname, a multi-proposer protocol that composes
  on top of a broad class of BFT consensus algorithms that determine
  decisions through explicit validator votes and support vote extensions:
  application-defined data that validators attach to their votes.

  \item \protoname achieves bounded inclusion: any \payload attested by
  more than $2f$ validators at height $h$ is finalized at height $h{+}1$.
  Combined with deterministic ordering, this bounds any single validator's
  discretion over transaction inclusion and ordering.

  \item \protoname decouples dissemination (\collectors) from agreement
  (validators), removes the need for a shared mempool, and eliminates the
  single-assembler egress bottleneck.

  \item We prove \protoname correct: it preserves the underlying
  protocol's safety and liveness properties,
  and bounded inclusion is guaranteed within one consensus height.
\end{enumerate}

The rest of the paper is organized as follows.
\Cref{sec:related} surveys related work,
\Cref{sec:model} presents the system model and building blocks,
\Cref{sec:protocol} describes \protoname in detail,
\Cref{sec:correctness} presents the correctness proofs,
and \Cref{sec:conclusion} concludes with directions for future work.

\section{Related Work}
\label{sec:related}

\textbf{Multi-leader protocols}, such as
Mir-BFT~\cite{stathakopoulou2019mirbft} and ISS~\cite{stathakopoulou2022iss},
distribute the dissemination and ordering load across multiple leaders.
Mir-BFT assigns portions of the transaction space to different leaders, while
ISS wraps a leader-driven ordering protocol in a
sequenced-broadcast abstraction and runs multiple consensus instances concurrently.
These designs improve throughput by removing the single leader's bandwidth
bottleneck, but they retain leader-driven ordering within each parallel
instance.
Partitioning transactions across leaders also introduces coordination overhead
and duplicate-detection requirements.
Furthermore, parallel leaders do not by themselves constrain the final block
assembler: a transaction assigned to a leader can still be delayed by that leader
or skipped by the merge mechanism.
In \protoname, a \payload attested by more than $2f$ validators must be included
in the next block.

\textbf{DAG-based protocols}~\cite{danezis2022narwhal, spiegelman2022bullshark,
babel2023mysticeti, spiegelman2023shoal, zarun2024shoalpp, bluebottle}
are a well-studied vein of multi-proposer solutions.
They allow validators to create blocks that reference blocks from previous
rounds and interpret the resulting DAG as a consensus protocol.
Although these protocols permit concurrent block production, their ordering
still relies on a leader-like role: an `anchor'.
Blocks in a committed anchor's history are included in the finalized sequence
of blocks.
If an anchor wants to censor a transaction from the previous round, it can omit
the block containing that transaction from its references. To mitigate this
issue, some DAG protocols allow multiple anchors per round. Anchors are ranked
within each round, and the committed sequence follows their ordered histories.
This still enables short-term censorship: if the highest-ranked anchor omits a
transaction known to all validators during the round, the transaction can only
be included in the block following that anchor's committed history.
\protoname constrains validators from analogous behavior.
Once a \payload is attested by more than $2f$ validators, every valid decision
references it, forcing its inclusion in the next block.

\textbf{Leaderless protocols}, such as
HoneyBadgerBFT~\cite{miller2016honeybadger},
DispersedLedger~\cite{yang2022dispersedledger}, and Prefix
Consensus~\cite{xiang2026prefix}, allow every node to contribute to each round's proposal
without relying on a designated leader.
In HoneyBadgerBFT, every node proposes a random sample of transactions,
and an asynchronous common subset protocol selects $n-f$ proposals to finalize.
DispersedLedger builds on this by decoupling agreement from block retrieval:
nodes agree on a log of block digests, then retrieve full blocks asynchronously
using erasure coding.
Prefix Consensus introduces a new primitive where nodes input transaction
vectors and agree on the maximum common prefix of all honest inputs; a
multi-slot partially synchronous construction improves censorship resistance
after Global Stabilization Time (GST).
None of these protocols, however, enforce bounded transaction inclusion:
HoneyBadgerBFT and DispersedLedger do not prevent censorship across rounds, and Prefix
Consensus permits up to $f$ slots to be censored even after GST.
\protoname instead enforces bounded inclusion: no \payload attested by more
than $2f$ validators can be deferred beyond the next block.
Moreover, \protoname allows non-validators to act as \collectors and to
broadcast \payloads.

\textbf{Autobahn}~\cite{giridharan2024autobahn} has every validator act as a
proposer.
Validators batch user transactions, collect acknowledgments from
other validators certifying storage, and chain the resulting proofs of
availability (PoAs); consensus then agrees on the tip of each validator's
chain.
Autobahn's main property, `seamlessness', stems from responsive transaction
dissemination and streamlined commit, achieved in four pre-consensus steps:
a user sends a transaction to a validator;
the validator batches and broadcasts it to all validators;
validators acknowledge that they have stored the batch;
and the original validator creates and disseminates a PoA.
The last step can be eliminated by increasing message complexity: validators
send acknowledgments to all validators, who then construct the PoA locally.
\protoname satisfies `seamlessness' in only two pre-consensus communication
steps: users send transactions to \collectors and \collectors send \payloads to
all validators.
Transactions propagate at the speed of the network, and any \payload
attested by more than $2f$ validators is finalized in the next block.

\textbf{Cadence}~\cite{cadence} is a leaderless multi-proposer framework
that composes an orchestrator with a single-shot consensus protocol.
The orchestrator determines when consensus slots (attempts to finalize a
block) occur.
At each slot, the single-shot consensus protocol is run with proposals possibly
being formed from every proposer, removing the single leader architecture.
Cadence operates in the partially synchronous model and guarantees
\payload inclusion only after a period of time following GST, because the
orchestrator must complete a recovery phase before slot timing is
communicated to proposers.
\protoname does not suffer the same issue, as it provides its bounded
inclusion guarantee regardless of network conditions.

\textbf{Solana} proposes Constellation~\cite{constellation}, a multi-proposer
protocol, to extend its blockchain with multiple concurrent proposers.
In it, users send transactions to a subset of validators acting as proposers
who erasure-code batches of transactions (\payloads in \protoname) and send
the fragments to another subset of validators acting as attesters.
Attesters send the fragments to the block assembler who waits to hear from
half of all attesters before assembling a block on which consensus is run.
Since only a quarter of fragments are necessary, the block assembler is able
to reconstruct full transactions from the attestations.
Proposers and attesters rotate amongst all validators.

Constellation inherits the fault model of Alpenglow~\cite{alpenglow}, its
underlying consensus protocol that tolerates up to $f < n/5$ Byzantine
validators (plus up to $p < n/5$ crashes).
Concretely, out of $\approx 1000$ validators, $\approx 16$ act as proposers
and $\approx 256$ as attesters.
This distribution is problematic: the block assembler requires responses from
half of all attesters.
If all Byzantine validators are attesters and do not send attestations to
the block assembler, then no blocks can be finalized and Alpenglow will
proceed at the speed of timeouts until enough correct validators are cycled
into attester roles.
The problem worsens when crashed validators are also present in the attester role.
In \protoname, \collectors disseminate \payloads to all validators, and any
quorum of $n-f$ validators can enforce inclusion, preventing Byzantine nodes
from concentrating in a role that stalls block production.

\textbf{Ethereum} addresses MEV through proposer--builder separation
(PBS)~\cite{buterin2021pbs}, which splits block construction into two roles:
builders assemble transactions into blocks and the consensus proposer selects
an assembled block via an auction.
This separates consensus participation from block construction, but does not
constrain the contents of the selected block.
The builder retains discretion
over transaction inclusion and ordering, while the proposer retains discretion
over which candidate block is selected.
PBS therefore only redistributes the control exercised by the single block
assembler to other actors.
Trusted relays~\cite{flashbots2022mevboost} and enshrined
PBS~\cite{buterin2023epbs} attempt to reduce this redistributed trust through
new mechanisms.

FOCIL~\cite{eip7805focil} takes a different approach: inclusion-list committee
members construct inclusion lists from their local mempools, and the block
assembler of the following height must include the listed transactions.
FOCIL provides stronger inclusion guarantees than PBS, but depends on the
formation, dissemination, and availability of a separate committee's lists.
In contrast, \protoname does not introduce a committee: \collectors disseminate
\payloads and the consensus quorum itself determines which \payloads every
valid decision must include.

\textbf{MCP}~\cite{garimidi2025mcp} (Multiple Concurrent Proposers) is a
protocol gadget that wraps any partially synchronous BFT consensus to achieve
selective-censorship resistance and hiding.
At each slot, multiple proposers encode transaction batches (\payloads in \protoname) using hiding
erasure-correcting codes and distribute shreds to a committee of relays.
Relays attest to the availability of received shreds without learning batch
contents, and the consensus leader must include sufficiently many relay
attestations to form a block, forcing the incorporation of honest proposers'
batches.
MCP achieves both selective-censorship resistance and hiding.
To additionally guarantee liveness, the fault tolerance must be reduced to
$f < n/5$; at the standard $f < n/3$ threshold, the guarantees are
probabilistic.
In MCP, the leader must either incorporate all honest batches or propose no
block at all for that slot.
In \protoname, each \payload that is attested by more than $2f$ validators must
be referenced on every valid commit certificate,
forcing per-\payload inclusion in the next block at standard fault tolerance.

\section{System Model}
\label{sec:model}

\protoname builds on two underlying components: a BFT consensus
protocol, referred to as the \emph{agreement layer} (\Cref{sec:consensus}),
and a broadcast abstraction, referred to as the \emph{dissemination layer}
(\Cref{sec:dissemination}).
\protoname inherits the system model of the agreement layer
and requires it to support vote extensions (\Cref{sec:extensions}).
\Cref{fig:architecture} illustrates the architecture of a validator node
and the interfaces between \protoname and both the agreement and dissemination layers.

The system consists of a set of processes, or \emph{nodes}, that communicate
by exchanging messages.
Nodes are subject to Byzantine failures. Faulty nodes may behave arbitrarily,
but cannot prevent correct nodes from communicating.
Correct nodes strictly follow the protocol.
Exchanged messages are authenticated using digital signatures under a
public-key infrastructure.
Each node holds a private signing key and the corresponding public keys are known
to all nodes.
Any node can verify the signature on a received message.

\protoname defines two node roles: validators and \collectors.
A node may play both roles:

\begin{itemize}
\item \textbf{Validators} execute the agreement layer and participate in
consensus.
For simplicity, we assume a fixed set of $n$ validators, of which at most $f$
may be faulty.
In practice, the validator set may change over heights, but not within
a height.
\protoname places no independent constraint on $n$ and $f$: the threshold is
determined by the agreement layer.
For simplicity, we assume equal voting power across validators; the protocol does not require it.

\item \textbf{\Collectors} mediate the interaction between users that submit
transactions to the system and validators that finalize blocks.
A \collector collects transactions from users,
bundles them into \payloads $p$,
and sends them to all validators via \textsc{Broadcast}($p$) on the dissemination layer.
Any node, including validators, can be assigned to the \collector role.
When describing \protoname, we assume all \collectors are honest.
The implications of faulty \collectors are discussed in
\Cref{sec:byz_proposers}.
\end{itemize}


\subsection{Agreement Layer}
\label{sec:consensus}

The agreement layer enables validators to agree on values.
We assume it implements a variant of BFT
consensus defined by the following properties:

\begin{itemize}
	\item \textbf{Agreement}: No two correct validators decide on different values.
	\item \textbf{Termination}: All correct validators eventually decide on a value.
	\item \textbf{Validity}: A decided value is valid, i.e., it satisfies the
		predefined \textsc{Validate}() predicate.
\end{itemize}

This BFT consensus variant adopts an application-specific \textsc{Validate}()
predicate to indicate whether a proposed value is
valid~\cite{cachin2001asynchronous};
only proposed values that are valid can be decided.
Agreement and Validity are \emph{safety} properties,
Termination is a \emph{liveness} property.

\protoname runs a sequence of consensus instances, uniquely identified by a
\emph{height}.
Each of them decides a value that \protoname maps into a block to be finalized
at the corresponding blockchain height $h$.
The basic interface offered by the agreement layer has three primitives:

\begin{itemize}
	\item \textsc{Propose}($h, v$): invoked by a validator to propose value $v$ at height $h$.
	By invoking this primitive, a validator starts consensus instance $h$.
	\item \textsc{Decide}($h, v, commit$): invoked by the agreement layer to
	notify of the decision of value $v$ at height $h$.
	$commit$ is a commit certificate: at least $n-f$ signed deciding
	votes for $v$, transferable and verifiable by any node.
	\protoname requires that a commit certificate be obtainable upon each decision.
	\item \textsc{Validate}($h, v$): implemented by \protoname and invoked by the
	agreement layer to check whether $v$ is a valid proposal at height $h$.
\end{itemize}

\protoname supports a broad class of BFT algorithms as the agreement layer,
including leader-based and leaderless protocols,
at $n > 3f$, $n > 5f$, and $n \geq 5f-1$ fault thresholds,
under partial synchrony, synchrony, or asynchrony.

\subsection{Vote Extensions}
\label{sec:extensions}

A vote extension is application-defined data that a validator attaches to a
consensus vote.
To the best of our knowledge, it was introduced by
CometBFT~\cite{cometbft_voteext_2026}, an implementation of Tendermint BFT
consensus protocol~\cite{buchman2018tendermint}.
It is adopted in Cosmos~\cite{buchman_kwon_2016} blockchains and
Polygon Heimdall~\cite{polygon_heimdall_2026}.
A similar concept also appears in recent theoretical work~\cite{multimmit}.

When the consensus protocol casts a $vote$ that may lead to a decision,
validators invoke the \textsc{ExtendVote}($vote$) primitive.
The application then returns a value $ext$, which the validator
attaches to the message carrying the $vote$, before it is signed and broadcast.
The extension is thus covered by the validator's signature and is
cryptographically bound to and disseminated with the $vote$.
For protocols requiring $n > 3f$ validators, extensions attach to the
votes of the final voting round; for protocols requiring $n > 5f$, to the first
and only voting round.
Protocols that decide after a variable number of
rounds~\cite{vonlanthen2024banyan, kudzu, hydrangea} must attach
extensions to every voting round that can establish a decision.

Upon receiving a message carrying a $vote$ that may lead to a
decision, the validator first verifies the message's signature.
If valid, the validator extracts from the message the associated vote extension
$ext$ and invokes the \textsc{VerifyExtension}($vote, ext$) primitive.
If the application rejects the vote extension, the validator drops the full
message, and the $vote$ is not delivered to the consensus protocol.

A validator decides at a height once it obtains decision-forming votes for the
same value from at least $n-f$ validators.
Upon the decision of $v$ at height $h$, the agreement layer invokes
\textsc{Decide}($h, v, commit$).
Two validators can decide the same $v$ at height $h$ by collecting a different
set of messages carrying deciding votes, and therefore producing distinct
commit certificates.
This is not an issue: both decisions are valid, and any node can check this by
verifying each commit certificate.
We assume a \textsc{ValidCommit}($h, commit$) method that verifies a commit
certificate for a height.
It checks the signature of each message, the unanimity of votes for the same
value, the cardinality of the set of voters (more than $n-f$), and the presence
of valid vote extensions.

Since at most $f$ validators are faulty, every commit certificate carries
extensions from at least $n-2f$ correct validators.
In \protoname, validators attest to received \payloads by including their
concise identifiers in vote extensions.
Thus, commit certificates represent a tamper-evident record of the \payloads
attested to by a quorum of validators.


\begin{figure}
    \begin{center}
    \begin{tikzpicture}[scale=0.9, every node/.style={transform shape},
      layer/.style={draw, rounded corners, minimum width=12.5cm, minimum height=1.2cm, align=center, font=\small},
      arrow/.style={-{Stealth[length=2.5mm]}, thick},
      biarrow/.style={{Stealth[length=2.5mm]}-{Stealth[length=2.5mm]}, thick},
      lbl/.style={font=\footnotesize, fill=white, inner sep=1pt}
    ]
    
    \node[layer, fill=blue!8] (app) at (0,4.5) {\textbf{Application}};
    \node[layer, fill=orange!12] (protocol) at (0,2.2) {\textbf{\protoname: Arc Multi-Proposer Protocol}};
    \node[layer, fill=green!8, minimum width=3cm] (dissemination) at (-4.75,-0.8) {\textbf{Dissemination}};
    \node[layer, fill=red!8, minimum width=9.3cm] (consensus) at (1.6,-0.8) {\textbf{Agreement: Consensus Algorithm}};
    
    \coordinate (carrtop) at (protocol.south -| consensus.north);
    \coordinate (carrbot) at (consensus.north);
    \draw[biarrow] (carrtop) -- (carrbot);
    \node[lbl, align=center] at ($(carrtop)!0.5!(carrbot)$) {
        \textsc{Validate}($h, v$) \\ \textsc{ExtendVote}($vote$) \\ \textsc{VerifyExtension}($vote, ext$)
    };
    \draw[arrow] ($(consensus.north east) + (-1.4,0)$) -- node[lbl, align=center] {\textsc{Decide}($h$, $v$, $commit$)} ($(protocol.south -| consensus.north east) + (-1.4,0)$);
    \draw[arrow] (dissemination.north) -- node[lbl, align=center] {\textsc{Deliver}($p_i$)} (protocol.south -| dissemination.north);
    \draw[biarrow] ($(protocol.north -| dissemination.north)$) -- node[lbl, align=center] {\textsc{ValidPayload}($p_i$)} ($(app.south -| dissemination.north)$);
    \draw[arrow] ($(protocol.north -| consensus.north east) + (-0.2,0)$) -- node[lbl, left=2pt] {\textsc{Finalize}($h, \textsc{Sort}(\{p_1, p_2, ...\})$)} ($(app.south -| consensus.north east) + (-0.2,0)$);
    \draw[arrow] ($(protocol.south -| consensus.north west) + (1.4,0)$) -- node[lbl, align=center] {\textsc{Propose}($h, v$)} ($(consensus.north west) + (1.4,0)$);
    \end{tikzpicture}
    \end{center}
    \caption{
    Architecture of a validator node.
    The \protoname logic sits between the application and the underlying
    dissemination and agreement layers.
    Each $p_i$ represents a \payload broadcast by a \collector and received
    by the validator.
    A validator attests to received \payloads by including their $ids$ in vote
    extensions.
    Consensus values are commit certificates carrying at least $n-f$ attestations
    to \payload $ids$.
    Upon a decision, \protoname finalizes the block assembled from sufficiently
    attested \payloads.
    }
    \label{fig:architecture}
\end{figure}

\subsection{Dissemination Layer}
\label{sec:dissemination}

The dissemination layer allows \collectors to send \payloads
to all validators.
It is defined in terms of a \textsc{Broadcast}($p$) primitive,
used by \collectors to submit a \payload $p$ to validators,
and a \textsc{Deliver}($p$) primitive,
which notifies a validator when \payload $p$ is received.
We consider a weak form of broadcast known as Best-Effort Broadcast, formally
defined by~\cite{cachin2011introduction}:

\begin{itemize}
	\item \textbf{Broadcast Validity}: If a correct sender node $s$
	(\collector) invokes \textsc{Broadcast}($p$),
	then every correct destination node (validator)
	eventually invokes \textsc{Deliver}($p$).

	\item \textbf{Integrity}: If a destination node (validator)
	invokes \textsc{Deliver}($p$) on a \payload $p$ from sender $s$ (\collector)
	and $s$ is a correct node,
	then $s$ had invoked \textsc{Broadcast}($p$).
\end{itemize}

Best-Effort Broadcast is an \emph{unreliable} broadcast primitive.
Delivery of \payloads is only guaranteed when the sender (a \collector) is
correct and potentially resends the same \payload multiple times.
If the sender of a \payload $p$ is faulty, $p$ may only be delivered to a
subset of correct destinations (validators).
This is particularly true when a \collector is Byzantine.
As a result, \protoname's correctness should not rely on dissemination-layer
guarantees, but on the properties of the agreement layer and the construction
of the multi-proposer protocol.

\section{Protocol}
\label{sec:protocol}

This section presents \protoname.
We first provide an overview of the protocol (\Cref{sec:overview}),
then describe the main steps of its operation (\Cref{sec:operation}),
provide some relevant implementation details and an alternative design (\Cref{sec:details}),
discuss the concern of malicious proposers (\Cref{sec:byz_proposers}),
and instantiate \protoname with existing BFT consensus algorithms (\Cref{sec:examples}).

\subsection{Overview}
\label{sec:overview}

\protoname is a multi-proposer protocol on top of a BFT consensus
algorithm that satisfies the requirements of \Cref{sec:consensus}.
In traditional blockchain protocols, the role of proposer consists of
(i) collecting user transactions (e.g., from the mempool),
(ii) assembling them into a block, and
(iii) proposing the assembled block for consensus.
In single-proposer protocols, step (iii) commonly combines both the
dissemination of and the agreement on the proposed block.

\protoname assigns those responsibilities to distinct roles.
A node acting as a \collector assembles user transactions into a
\payload and disseminates it to all validators (dissemination).
Each \payload is uniquely identified by a concise identifier, its $id$.
A node acting as a validator runs the underlying BFT consensus algorithm
to decide, at each height, which \payloads to finalize (agreement).
A validator attests to received \payloads by including the \payload
$ids$ in signed vote extensions.
Consensus values are commit certificates containing at least $n-f$
attestations to \payload $ids$.
Vote extensions thus bridge dissemination and agreement.

When the agreement layer decides, the decided value is a commit certificate
from which \protoname extracts the \payload $ids$ attested by more than $f$
validators and maps them to the corresponding full \payloads.
\protoname then orders those \payloads with a deterministic function,
assembles the block for the decided height, and finalizes it.
Under normal operation, validators have already received the
\payloads corresponding to the \payload $ids$ extracted from the decision.
If a validator misses a \payload, the validator requests its retransmission.
Since each was attested by more than $f$ validators, at least one correct
validator holds the full \payload, guaranteeing availability.

A valid consensus value for height $h+1$ is a commit certificate for height
$h$, carrying vote extensions of at least $n-f$ validators.
A validator may omit at most $f$ attestations.
Consequently, a \payload $id$ attested by more than $2f$ validators at height
$h$ retains more than $f$ attestations in every valid consensus value for
height $h+1$, and is necessarily finalized there.
This is \protoname's \emph{bounded inclusion} guarantee: no validator, at any
height, can defer a \payload attested by more than $2f$ validators beyond the next block.

\begin{algorithm}
\caption{\protoname pseudo-code at a validator}
\begin{algorithmic}[1]

\Init{}
	\State $\V{ordered}, \, \V{payloads} \gets nil$ \Comment{Maps keys $\rightarrow$ values}
	\State $\V{pending} \gets \emptyset, \, \V{height} \gets 1$ \Comment{The next height to finalize}
\EndInit

\Space

\Event{\textsc{Deliver}}{$\payload$}{} \label{algo:deliver}
    \State $id \gets \textsc{ID}(payload)$ \label{algo:id}
	\If{$\V{payloads}[id] = \emptyset \: \land$ \textsc{ValidPayload}($payload$)} \label{algo:valid}
		\State $\V{payloads}[id] \gets payload$
		\If{$id \notin \V{ordered}[*]$} \Comment{Fresh \payload from \collector}
		\State $\V{pending} \gets \V{pending} \cup \{id\}$ \label{algo:deliverend}
        \EndIf
	\EndIf
\EndEvent

\Space

\Event{\textsc{Validate}}{$h, v$}{} \label{algo:validate}
	\Return \textsc{ValidCommit}($h-1, v$) \label{algo:validCommit}
\EndEvent

\Space

\Event{\textsc{ExtendVote}}{\texttt{VOTE}$\langle h, v, \ldots\rangle$}{} \Comment{Return a set of \payload $ids$} \label{algo:ves}
	\Return $\V{pending} \setminus \textsc{ExtractIDs}(v)$ \Comment{Skip $ids$ already ordered by $v$} \label{algo:vesend}
\EndEvent

\Space

\Event{\textsc{VerifyExtension}}{\texttt{VOTE}$\langle h, v, \ldots\rangle, attested$}{} \Comment{Attested $ids$ are pending} \label{algo:verify}
	\Return $attested \cap \V{ordered}[*] = \emptyset \land attested \cap \textsc{ExtractIDs}(v) = \emptyset$
\EndEvent

\Space

\Event{\textsc{Decide}}{$h, v, commit$}{} \label{algo:decide}
	\State $\V{ordered}[h] \gets \textsc{ExtractIDs}(v)$ \label{algo:ordered}
	\State $\V{pending} \gets \V{pending} \setminus \V{ordered}[h]$
	\State $\textsc{Propose}(h+1, commit)$ \Comment{Commit of height $h$ is the input for height $h+1$} \label{algo:propose}
\EndEvent

\Space

\Cond{$\V{ordered}[\V{height}] \neq nil \land (\forall id \in \V{ordered}[\V{height}] : \V{payloads}[id] \neq \emptyset)$} \label{algo:finalize}
	\State $\V{orderedPayloads} \gets \{\V{payloads}[id] \mid  id \in \V{ordered}[\V{height}] \}$\label{algo:extract}
	\Trigger{\textsc{Finalized}}{$\V{height}$, $\textsc{Sort}(\V{orderedPayloads})$} \label{algo:sort}
	\State $\V{height} \gets \V{height} + 1$
\EndCond

\Space

\Function{ExtractIDs}{$commit$} \Comment{Return \payload $ids$ attested by $>f$ validators}
	\Return $\bigl\{\, id :
		\lvert\{(vote, attested) \in commit : id \in attested\}\rvert > f \,
	\bigr\}$
	\label{algo:threshold}
\EndFunction

\end{algorithmic}
\label{algo:validator}
\end{algorithm}

\subsection{Operation}
\label{sec:operation}

\protoname mediates between the application, on one hand,
and the dissemination and agreement layers, on the other hand.
\Cref{fig:architecture} illustrates the design and \Cref{algo:validator} 
shows pseudo-code for a validator.
\protoname exposes two primitives to the application: \textsc{ValidPayload}($p$),
invoked to determine whether a received \payload is valid, and
\textsc{Finalize}($h, block$), invoked when a block is finalized at height $h$.
The main steps of the protocol operation follow below:

\begin{enumerate}
	\item \textbf{Collection.}
		Every \collector collects user transactions and assembles them
		into \payloads.
		The assembling policy (e.g., how large are \payloads, how long to wait
		to fill a \payload) is abstracted from this description and not
		represented in the pseudo-code.

	\item \textbf{Dissemination.}
		Once a \collector assembles a \payload $p$, it sends $p$ to all
		validators
		using the \textsc{Broadcast}($p$) primitive
		of the dissemination layer.

	\item \textbf{Payload Validation.}
		When a validator receives a \payload from the dissemination
		layer, via the \textsc{Deliver}($p$) primitive (line~\ref{algo:deliver}), it
		invokes \textsc{ValidPayload}($p$).
		The validator drops invalid \payloads.

	\item \textbf{Payload Identifier.}
		When a valid \payload is received, the validator assigns it a concise $id$.
		The full \payload is stored in the $\V{payloads}$ map, indexed by its $id$,
		while the \payload $id$ is added to the $\V{pending}$ set, containing $ids$
		of received but not yet ordered \payloads.
		\protoname orders \payloads by running consensus over \payload
		$ids$, not full \payloads. 

		We assume the \textsc{ID}($p$) method (line~\ref{algo:id}) produces a
		unique identifier for each \payload.
		A minimal implementation of it can rely on a
		collision-resistant hash function.

	\item \label{item:attestation} \textbf{Attestation.}
		When the agreement layer produces a $vote$ message that may
		lead to a decision, the validator invokes the
		\textsc{ExtendVote}($vote$) primitive.
		This primitive returns a vote extension to be attached to the
		$vote$ message before it is signed and broadcast.
		A validator attests to received \payloads by including their
		$ids$ in vote extensions.
		The $ids$ of valid \payloads received but not yet ordered are stored in
		the $\V{pending}$ set.
		From this set, the validator removes the \payload $ids$ extracted from
		the value $v$ the $vote$ is for (line~\ref{algo:vesend}), thus
		preventing the same $ids$ from being ordered in more than one
		height.

	\item \textbf{Vote Extensions.}
		When a validator receives a $vote$ message that may lead to a decision,
		and therefore includes a vote extension $ext$,
		it invokes the \textsc{VerifyExtension}($vote, ext$)
		primitive (line~\ref{algo:verify}).
		This allows validators to verify each other's vote extensions, which
		in \protoname are attestations to \payload $ids$.
		A vote extension is valid if all included \payload $ids$ are 
		still pending and are not part of the $ids$ extracted from the value
		$v$ the $vote$ is for.
		If \textsc{VerifyExtension}($vote, ext$) rejects the vote extension,
		then the $vote$ message itself is invalid: the message is dropped,
		instead of being delivered to the agreement layer.

	\item \label{item:consensus-value} \textbf{Consensus Value.}
		The consensus value is the object validators agree on at each
		height.
		In \protoname, it is a commit certificate: a quorum of at least
		$n-f$ signed votes from distinct validators for the same value,
		proving that the agreement layer decided that value.
		Each vote carries a vote extension, which is an attestation to a set
		of \payload $ids$ (item~\ref{item:attestation}).

	\item \textbf{Decision.}
		When the agreement layer decides value $v$ at height $h$
		(line~\ref{algo:decide}), the validator extracts from $v$ a set of \payload
		$ids$ attested by more than $f$ validators.
		This set of \payload $ids$ is stored in the $\V{ordered}[h]$ set and
		removed from the $\V{pending}$ set.

	\item \textbf{Finalization.}
		Under normal operation, the \payloads referenced in $\V{ordered}[h]$ are
		already available.
		If any are missing, the validator requests their retransmission.
		Since each \payload $id$ was attested by more than $f$ validators, at
		least one correct validator holds the full \payload and can retransmit
		it.
		Once all \payloads are available (line~\ref{algo:finalize}), the validator
		orders them deterministically via \textsc{Sort}($\V{orderedPayloads}$).
		The resulting block is a list of transactions drawn from the ordered \payloads.
		The validator then invokes \textsc{Finalize}($h, block$), delivering the block to the application.

	\item \textbf{Proposal.}
		When the agreement layer decides a value at height $h$
		(line~\ref{algo:decide}), it delivers to \protoname a $commit$
		certificate certifying that decision.
		Each validator immediately enters the next height $h+1$ using this
		certificate as its proposed consensus value (line~\ref{algo:propose}).

	\item \textbf{Proposal Validation.}
		When the agreement layer receives a proposed value, it invokes the
		external \textsc{Validate}($h, v$) predicate.
		\protoname accepts $v$ at height $h$ if and only if
		\textsc{ValidCommit}($h-1, v$) accepts $v$ as a commit certificate for
		the previous height (line~\ref{algo:validCommit}).
		By the validity property of the agreement layer, only values accepted by this
		predicate can be decided, ensuring that every decided value is a
		well-formed consensus value (item~\ref{item:consensus-value}).
		A special case arises at the first height:
		since there is no commit certificate,
		\textsc{ValidCommit}($h_0, v$) accepts a distinguished empty
		certificate for the genesis height $h_0$.

\end{enumerate}

\subsection{Implementation Details}
\label{sec:details}

\subsubsection{Transaction Ordering \& Deduplication}
\label{sec:order}

Once a set of \payload $ids$ is decided for a height and the full \payloads are
available (line~\ref{algo:finalize}), all correct validators must apply the
same logic to produce the same finalized block.
This means transforming a set of \payloads, each holding user
transactions, into a single totally-ordered sequence of transactions.
\protoname delegates this logic to the \textsc{Sort}($\V{orderedPayloads}$) method
(line~\ref{algo:sort}), which must be deterministic.
The specific ordering policy, however, is defined at the application level.
We recommend sorting transactions across all finalized \payloads in descending
order of priority fee per unit of computation, breaking ties in a pseudo-random manner.
This ordering is manipulation-resistant and market-based, as
users who value earlier execution pay higher fees.

Users with latency-sensitive transactions may send them to multiple
\collectors to reduce ordering latency.
Since duplicate transactions should not be sent to the application after finalization, we require \textsc{Sort}($\V{orderedPayloads}$) to deduplicate transactions.
The method should only keep the transaction copy ranked highest, removing all other copies and pseudo-randomly selecting one copy in the case of ties.
While out of scope of this work, spam should be discouraged by requiring users to pay fees for each proposer they send their transaction to.

\subsubsection{Execution}

A validator entering height $h+1$ proposes the commit certificate for height
$h$, which may reference \payloads it has not yet received.
It can therefore neither execute those \payloads nor commit to the state that
executing them produces.
This `asynchronous execution'~\cite{yin2003} may not be desirable and
indeed, three major blockchains--- Bitcoin, Ethereum, and Solana---do not use
asynchronous execution.  Instead, they pre-execute blocks before ordering.

It is possible to remove asynchronous execution in \protoname.
One option is to have validators wait to receive all \payloads attested by more
than $f$ validators in commit certificates before proposing the certificate for the next height.
The validators can then assemble the block corresponding to that proposal, using the same
\textsc{Sort}($\V{orderedPayloads}$) function, and execute its transactions.
However, adding retrieval on the critical path means consensus progresses at the pace of the slowest fetch, which may hurt performance.
Synchronous execution is better recovered by making the certificate
self-contained, so that whoever proposes can reconstruct \payload contents
locally.

\textbf{Alternative design.}
Currently, \collectors send \payloads with full transactions to all validators
who include \payload $ids$ in vote extensions.
To avoid asynchronous execution, \collectors should instead erasure code
\payloads and send the fragments to validators (similar to~\cite{constellation}).
Validators then include received fragments in vote extensions.
The erasure coding should be configured
so that more than $f$ fragments are needed to recover the full \payload.
With this change, validators are able to reconstruct all \payloads that can be
ordered in a height.
They can then execute the transactions of the block that would be finalized if
the commit certificate they propose via the agreement layer is decided.
The state resulting from the execution of this block could also be included in
the proposal, and verified by other validators.
A downside of this design is that vote extensions carrying erasure coded
fragments will be significantly larger than when carrying \payload $ids$.
Furthermore, validators are unable to check the validity of \payloads until
they are able to rebuild them from their fragments.

\textbf{Optimistic Execution.}
 
A validator can begin executing a \payload before the decision of the height
$h$ that would order it.
Once a validator observes, among the height $h$ votes it has collected that may
lead to a decision, more than $2f$ attestations to a \payload,
\protoname guarantees that the \payload will be finalized at height $h+1$.
The validator can therefore immediately execute the transactions included in
such \payload, rather than waiting to learn the other \payloads that will also
be finalized at height $h+1$.
This approach deliberately departs from the ordering policy of \Cref{sec:order}, where the finalized block is produced by sorting transactions across all \payloads ordered at a height.
It is best suited for transactions that do not contend for shared state.
Transfers between single-owner accounts, for instance, produce the same result regardless of their position in the sorted block~\cite{guerraoui2019consensusnumber}, so an optimistically executed result can be kept without revision.
Transactions that access shared state depend on execution order and must await the fully sorted block.

\subsection{Malicious Proposers}
\label{sec:byz_proposers}

While \protoname is described under the assumption of honest \collectors, it is
valuable to consider what happens if some \collectors are Byzantine.
The first consequence is the production and dissemination of
\payloads at a very high rate to exhaust the resources of validators.
In this case, it is up to the application layer, via the
\textsc{ValidPayload}($p$) primitive, to identify and drop \payloads that do not
contain unique and valid transactions.
Practical setups of \protoname should enable the application layer to disregard
\payloads from \collectors identified as malicious.
In other multi-proposer solutions, flooding validators with \payloads can be
more impactful.
For example, in Autobahn~\cite{giridharan2024autobahn}, every \payload has to be
acknowledged by more than $f$ validators before being ordered,
meaning that correct validators must verify at least $f$ signatures per fraudulent
\payload before rejecting it.
In \protoname, the \textsc{ValidPayload}($p$) primitive can drop invalid
\payloads at their reception.

A second avenue of attack consists of the production of valid \payloads that
are selectively disseminated so they only reach a subset of validators.
There are two scenarios to consider based on whether such a \payload reaches more than $f$ correct validators.
If it does not, the attack consists of sending \payloads that are attested by
correct validators but possibly never cross the threshold of $f$ attestations.
The $\V{pending}$ set and $\V{payloads}$ map grow in size, as do vote
extensions; but this does not result in more \payloads being ordered.
Although this form of attack cannot be fully prevented, its consequences can be
attenuated.
Validators can have policies restricting the size of their $\V{pending}$ sets, so
that a \payload received during height $h$ is removed from the $\V{pending}$
set, and thereafter ignored, if not ordered by height $h+c$, where $c$
is a constant.
Correct validators could also exchange $ids$ of \payloads that do not receive
enough attestations, rebroadcasting them on-demand.
Another approach is to rely on a stronger, but also costlier, form of
broadcast, such as reliable broadcast~\cite{cachin2011introduction}.

Even if a \payload is attested by more than $f$ validators, reaching the
availability threshold, the selective dissemination by malicious \collectors
can have negative impact.
If a \payload is ordered, but only a single correct validator holds it, while
the other $f$ faulty validators do not retransmit it, finalization will require
a significant number of retransmission requests and responses.
This attack is hard to distinguish from legitimate unreliable communication.
As previously mentioned, a dissemination layer with stronger guarantees should
minimize the impact of such an attack---the cost of which is higher message
complexity in the common case.
Another approach is to build more robust \payload retransmissions protocols,
to reduce duplication and balance the load of re-broadcasts among all
validators.

Byzantine \collectors may also tamper with payload contents when assembling \payloads.
First, a Byzantine \collector may selectively exclude transactions from its
\payloads.
A user can mitigate this by submitting
transactions to multiple \collectors. 
As long as at least one honest \collector
includes the transaction in a \payload that reaches sufficient validators,
bounded inclusion ensures it is finalized in the next block.
Submitting to multiple \collectors, however, may increase fees as the user will need to pay for each
\collector that includes their transaction.
Second, a Byzantine \collector may insert transactions into its \payload to front-run or sandwich a victim transaction submitted by a user.
To place an attack transaction immediately before or after the victim, the attacker assigns it a priority fee marginally above or below the victim's (assuming the ordering policy of \Cref{sec:order}).
The attack succeeds only with some probability: transactions from other finalized \payloads, whose contents are unknown to the attacker at construction time, may carry fees that fall between the attack and victim, causing the deterministic ordering to insert them in between and nullify the attack.
To increase the probability of success, the attacker must narrow the fee gap between attack and victim transactions.
A smaller gap, however, makes the manipulation more easily detectable.
Without user-side transaction encryption, this attack applies to all multi-proposer protocols~\cite{giridharan2024autobahn,garimidi2025mcp,constellation,cadence}.

More broadly, since any node may act as a \collector in \protoname, \collectors compete for user transactions.
They therefore have an incentive to maintain their reputation: a
\collector known to drop or attack transactions will lose users to
competitors.

\subsection{Instantiations}
\label{sec:examples}

The preceding sections describe \protoname over a generic agreement layer
satisfying the requirements of \Cref{sec:consensus}.
The following examples show how \protoname composes with existing BFT consensus
algorithms, identifying in each case which message carries vote extensions and
any protocol-specific adaptations the composition entails.
\protoname attaches vote extensions to the final voting round and uses a quorum
of such votes as the commit certificate.

\textbf{PBFT}~\cite{castro1999pbft}
assumes $n > 3f$ and operates in three communication steps, each one with its
corresponding message:
\textsc{pre-prepare} (primary broadcasts the full proposal),
\textsc{prepare} (validators accept the proposal),
and \textsc{commit} (validators commit to a prepared proposal).
A validator decides after collecting $2f+1$ matching \textsc{commit} messages.
\protoname composes with PBFT by attaching vote extensions to
\textsc{commit} messages that lead to decisions.

\textsc{prepare} and \textsc{commit} messages carry a proposal's digest
instead of the full proposal.
To recover the full proposal, validators have to map digests to the full
proposals received in \textsc{pre-prepare} messages.
This is a requirement for the proper operation of \textsc{ExtendVote}()
(line~\ref{algo:ves}), which requires the full value $v$ the validator commits
for, not only $v$'s digest.

PBFT predates blockchain systems and supports concurrent consensus instances,
each with distinct sequence numbers---which map directly to heights.
\protoname, however, requires lock-step operation: height $h+1$ begins only
after height $h$ is decided, ensuring the commit certificate for height $h$
is available as the proposed value at height $h+1$.

PBFT's view change mechanism allows replicas to replace a suspected primary
(leader).
It involves the exchange of \textsc{view-change} messages containing proofs of
the values prepared in the current view: one \textsc{pre-prepare} and $2f$
\textsc{prepare} messages per height.
Since \protoname attaches vote extensions to \textsc{commit} messages,
\textsc{view-change} messages require no modification.
Any undecided height that carries over into the new view will re-enter
the \textsc{prepare} phase; validators attach fresh vote extensions when they
issue \textsc{commit} messages in the new view.

\textbf{Tendermint}~\cite{buchman2018tendermint}
also assumes $n > 3f$ and operates in three communication steps, each one with
its corresponding message:
\textsc{proposal} (leader broadcasts the proposed value),
\textsc{prevote} (validators accept the proposal),
and \textsc{precommit} (validators lock and commit to a proposal).
A validator decides after collecting $2f+1$ matching \textsc{precommit}
messages.
\protoname composes with Tendermint by attaching vote extensions to
\textsc{precommit} messages.

As in PBFT, Tendermint's \textsc{prevote} and \textsc{precommit} messages
carry a digest of $v$, not the full proposal.
A validator must first receive a \textsc{proposal} for $v$ before
\textsc{ExtendVote}() can be invoked for a \textsc{precommit} referencing $v$.
Tendermint's vote extension interface already handles this: \textsc{ExtendVote}()
is invoked only after \textsc{ProcessProposal}($v$) has exposed the full
proposal to the application~\cite{cometbft_voteext_2026}.
In \protoname, the \textsc{Validate}($h, v$) predicate
(line~\ref{algo:validate}) plays the same role as
\textsc{ProcessProposal}($v$): an application-specific validation function.

Unlike PBFT, Tendermint rotates leaders as part of its ordinary
operation.
Each height has a different validator as its leader and there is no explicit
view change mechanism.
Agreement is guaranteed by a locking mechanism that is based on the reception
of a \textsc{proposal} and $2f+1$ matching \textsc{prevote} messages.
Since vote extensions are attached to \textsc{precommit} messages, no changes
are required in the implementation of the consensus algorithm.

Fast Tendermint~\cite{fasttendermint} is a recent Tendermint variation
that assumes $n > 5f$ and reaches consensus in a single voting round.
Since it preserves the \textsc{precommit} message as the deciding vote, while
dropping the \textsc{prevote} step, it also composes with \protoname without modification.

\textbf{HotStuff}~\cite{yin2019hotstuff}
assumes $n > 3f$ and achieves linear communication complexity by routing all
votes through the leader.
It operates in three voting rounds, each producing a quorum certificate (QC)
of $n-f$ votes:
\textsc{prepare} (validators vote on the leader's proposal),
\textsc{pre-commit} (validators vote on the \textsc{prepare} QC),
and \textsc{commit} (validators vote on the \textsc{pre-commit} QC).
The leader assembles a \textsc{commit} QC from $n-f$ \textsc{commit}
votes and broadcasts it in a \textsc{decide} message.
\protoname can compose with HotStuff by attaching vote extensions to any
message, since the leader relays the resulting QC to all validators.
For the original HotStuff, the best candidates are \textsc{commit} votes
that are included in the leader's \textsc{decide} message.

HotStuff rotates leaders across heights and views.
The incoming leader collects \textsc{new-view} messages, each carrying the
highest \textsc{prepare} QC known to that replica.
Since \protoname attaches vote extensions only to \textsc{commit} votes,
\textsc{new-view} messages require no modification.

Chained HotStuff pipelines the three voting rounds:
each vote simultaneously acts as a \textsc{prepare} for the current height,
a \textsc{pre-commit} for the previous height,
and a \textsc{commit} for the height two rounds back.
For \protoname, this means that validators include attestations in every vote
message, as the vote plays the role of a \textsc{commit} to a previously
proposed value.
The next leader's message includes a QC of $n-f$ such votes as the proposed
value.

\textbf{Simplex}~\cite{chan2023simplex}
assumes $n > 3f$ and operates in two voting rounds per height.
The leader broadcasts a \textsc{propose} message carrying the proposed value
and a certificate from the previous height.
Validators cast a \textsc{vote} for the proposal; $n-f$ votes constitute a
certificate.
Upon collecting $n-f$ votes, validators broadcast a \textsc{finalize} message;
$n-f$ \textsc{finalize} messages establish a decision.
\protoname composes with Simplex by attaching vote extensions to
\textsc{finalize} messages---the final voting round.
The commit certificate is a quorum of $n-f$ \textsc{finalize} messages,
each carrying a vote extension with the validator's \payload attestations.

The leader's \textsc{propose} message carries the full proposed value.
Validators store it upon receipt, so the full value is available when
\textsc{ExtendVote}() is invoked for \textsc{finalize} messages.

Simplex's view change requires no dedicated sub-protocol: validators broadcast
a \textsc{view-change} message upon timeout, and $n-f$ \textsc{view-change}
messages for view $k$ prove that no decision occurred in that view, triggering
entry into view $k+1$.
Since \protoname attaches vote extensions only to \textsc{finalize} messages,
\textsc{view-change} messages require no modification.

\section{Correctness}
\label{sec:correctness}

This section demonstrates the correctness of \protoname,
based on the properties of the agreement layer,
defined in \Cref{sec:consensus}:
Agreement (Theorem~\ref{thm:agreement}),
Validity (Theorem~\ref{thm:validity}),
and Termination (Theorem~\ref{thm:termination}).
\protoname inherits these properties from the underlying BFT consensus algorithm.
Then we formalize and prove \protoname's key guarantee: Bounded Inclusion (Theorem~\ref{thm:bounded}).


\begin{theorem}\label{thm:agreement}
	No two correct validators finalize different values at the same height.
\end{theorem}

\begin{proof}
By the Agreement property of the agreement layer, no two correct validators
decide different values at the same height.
Therefore, when the \textsc{Decide}($h, v, commit$) primitive is invoked
(line~\ref{algo:decide}) on any correct validator, it carries the same
value $v$ for height $h$.
Every validator that decides $v$ at height $h$
runs the deterministic function \textsc{ExtractIDs}($v$)
(line~\ref{algo:ordered}).
As a result, every correct validator produces the same set $\V{ordered}[h]$ of
\payload identifiers and computes the same set of \payloads
$\V{orderedPayloads}$ for height $h$ (line~\ref{algo:extract}), in particular
because we assume \textsc{ID}($p$) to be collision-resistant.
All the transactions bundled into these \payloads are sorted and assembled into
a block, by the deterministic function \textsc{Sort}($\V{orderedPayloads}$),
and this block is finalized for the height $h$ (line~\ref{algo:sort}).
As a result, every correct validator that finalizes height $h$ finalizes the
same block, assembled from the \payloads referenced in $\V{ordered}[h]$.
\end{proof}


\begin{theorem}\label{thm:validity}
	Every value finalized by a correct validator is valid.
\end{theorem}

\begin{proof}
The value finalized in a height $h$ is a block produced by
\textsc{Sort}($\V{orderedPayloads}$) (line~\ref{algo:sort}).
The validity of the finalized block follows from the validity of each \payload
in $\V{orderedPayloads}$.
Let $p$ be any \payload in $\V{orderedPayloads}$.
By line~\ref{algo:ordered}, \textsc{ID}($p$) was part of the set returned by
\textsc{ExtractIDs}($v$).
Hence, by line~\ref{algo:threshold}, more than $f$ validators attested to $p$.
Since at most $f$ validators are Byzantine, at least one such attestation was
produced by a correct validator.
A correct validator attests to \payloads whose identifiers are stored in the
$\V{pending}$ set (line~\ref{algo:vesend}).
\textsc{ID}($p$) is added to $\V{pending}$ only if \textsc{ValidPayload}($p$)
returned true (line~\ref{algo:valid}).
Thus, the application running in a correct validator has considered $p$ valid.
\end{proof}


\begin{theorem}\label{thm:termination}
	All correct validators eventually finalize a value for each height.
\end{theorem}

\begin{proof}
The Termination property of the agreement layer ensures that every correct
validator eventually decides a value for each height $h$.
When \textsc{Decide}($h, v, commit$) is invoked, each correct validator sets
$\V{ordered}[h]$ to the output of $\textsc{ExtractIDs}(v)$ (line~\ref{algo:ordered}).
This method extracts from $v$ a set of \payloads, referenced by their
identifiers, that were attested by more than $f$ validators
(line~\ref{algo:threshold}).

Since at most $f$ validators are Byzantine, at least one correct validator has
attested to each such \payload $p$.
A correct validator attests to $p$ (lines~\ref{algo:ves}--\ref{algo:vesend})
only after receiving and validating $p$ (lines~\ref{algo:deliver}--\ref{algo:deliverend}).
Any correct validator missing $p$ can, upon receiving the commit certificate,
identify the attesters and request retransmission of $p$.
Since at least one attester is correct and holds $p$, retransmission eventually
succeeds.
Thus, the condition on line~\ref{algo:finalize} eventually holds for every
correct validator and each height $h$, and the content of all \payloads in
$\V{ordered}[h]$ is finalized (line~\ref{algo:sort}).
\end{proof}


\begin{theorem}\label{thm:bounded}
	A \payload attested by more than $2f$ validators at height $h$ is finalized at height $h+1$.
\end{theorem}

\begin{proof}
%
\protoname requires the value decided by the agreement layer at height $h+1$ to
be a valid commit certificate for height $h$.
A valid commit certificate for height $h$ must include signed messages from at
least $n-f$ distinct validators running height $h$, containing their vote
extensions, produced during height $h$.
\protoname enforces this through the \textsc{Validate}($h+1, v$)
primitive (lines~\ref{algo:validate}--\ref{algo:validCommit}), which defines
the validity of consensus values.

Let $p$ be a \payload attested by more than $2f$ validators at height $h$.
Since vote extensions from at most $f$ validators can be omitted from a valid
commit certificate for height $h$, more than $f$ attestations to $p$ must be
included in any valid commit certificate for height $h$.
Consequently, when \textsc{Decide}($h+1, v, commit$) is invoked,
the output of \textsc{ExtractIDs}($v$) must include $\textsc{ID}(p)$
(line~\ref{algo:threshold}), for any value $v$ that can be decided by the
agreement layer at height $h+1$.
Following the same rationale as in Theorem~\ref{thm:termination},
$\V{ordered}[h+1]$ does include $\textsc{ID}(p)$,
and the content of $p$ is eventually finalized by any correct validator at
height $h+1$. 
\end{proof}

\section{Conclusion}
\label{sec:conclusion}

\protoname is a multi-proposer protocol that composes with a broad class
of BFT consensus algorithms, allowing any node to act as a \collector and
participate in block construction.
\Collectors disseminate \payloads to validators, who attest to received
\payloads by attaching \payload identifiers to their consensus votes.
This enforces \emph{bounded inclusion}: any \payload attested by more than $2f$
validators at height $h$ is finalized at height $h+1$, regardless of network
conditions.
Combined with deterministic transaction ordering, \emph{bounded inclusion} constrains
any single validator's power to selectively exclude or reorder transactions.

Future work includes implementing \protoname and experimentally evaluating the
performance impacts from decoupling dissemination from agreement.
Second, \protoname delegates dynamic validator set membership to the underlying
agreement layer, which requires an availability-preserving handoff for incoming
validators covering all ordered but not yet finalized \payloads.
Third, extending \protoname with pre-finalization transaction confidentiality
is a natural extension.
Since \protoname already disseminates each \payload to all validators and requires
multiple attestations before finalization, its design is well suited to
threshold encryption.

\vskip\baselineskip\noindent
\acknowledgementsHeading
We would like to thank Circle engineers Nenad Milosevic, Sergio Mena, and
Alessandro Sforzin for discussions regarding \protoname.

\bibliographystyle{plainurl}
\bibliography{main}

\end{document}